\documentclass[journal]{IEEEtran}
\usepackage{color}
\usepackage[cmex10]{amsmath}
\usepackage{graphicx}
\usepackage{amsthm}
\usepackage{amssymb}
\usepackage{multicol}
\usepackage{cite}
\usepackage{algpseudocode}
\usepackage{algorithm}
\usepackage{amsfonts}
\usepackage{latexsym,psfrag}
\usepackage{makeidx}
\usepackage{todonotes}
\usepackage{comment}
\usepackage{dsfont}
\usepackage{multirow}
\usepackage{url}
\usepackage{xcolor}
\usepackage{booktabs}
\usepackage{ifpdf}
\usepackage{soul}
\usepackage{caption}
\DeclareCaptionFont{9.5pt}{\fontsize{9.5pt}{9.5pt}\selectfont}
\newtheorem{theorem}{Theorem}

\newtheorem{claim}{Claim}

\def\be{\begin{equation}}
\def\ee{\end{equation}}

\def\bearn{\begin{eqnarray*}}
\def\eearn{\end{eqnarray*}}
\def\bear{\begin{eqnarray}}
\def\eear{\end{eqnarray}}
\def\barr{\begin{array}}
\def\earr{\end{array}}

\def\bmat{\left(\begin{array}}
\def\emat{\end{array}\right)}

\usepackage{etoolbox}\AtBeginEnvironment{algorithmic}{\small}

\begin{document}

\title{TDOA Source-Localization Technique Robust to Timing Attacks}
\author{Marguerite Delcourt \textit{Student Member, IEEE}, Jean-Yves Le Boudec \textit{Fellow, IEEE}
\thanks{M. Delcourt and J-Y. Le Boudec are with the School of Computer and Communication Sciences of the Swiss Federal Institute of Technology Lausanne, EPFL, Switzerland. e-mail: \{marguerite.delcourt$\vert$ jean-yves.leboudec\}@epfl.ch.}}
\maketitle

\begin{abstract}
In this paper, we focus on the localization of a passive source from time difference of arrival (TDOA) measurements. TDOA values are computed with respect to pairs of fixed sensors that are required to be accurately time-synchronized. This constitutes a weakness as all synchronization techniques are vulnerable to delay injections. Attackers are able either to spoof the signal or to inject asymmetric delays in the communication channel. By nature, TDOA measurements are highly sensitive to time-synchronization offsets between sensors. Our first contribution is to show that timing attacks can severely affect the localization process. With a delay of a few microseconds injected on one sensor, the resulting estimate might be several kilometers away from the true location of the unknown source. We also show that residual analysis does not enable the detection and identification of timing attacks. Our second contribution is to propose a two-step TDOA-localization technique that is robust against timing attacks. It uses a known source to define a weight for each pair of sensors, reflecting the confidence in their time synchronization. Our solution then uses the weighted least-squares estimator with the newly created weights and the TDOA measurements received from the unknown source. As a result, our method either identifies the network as being too corrupt to localize, or gives a corrected estimate of the unknown position along with a confidence metric. Numerical results illustrate the performance of our technique.

\end{abstract}

\section{Introduction}

The problem of localizing uncooperative sources that emit radio frequency signals has been extensively studied in the field of electronic warfare~\cite{Poisel12}, as well as for civil applications~\cite{LYJZDGW19}. Solutions were proposed in various settings such as sensor networks, radar, sonar or wireless communication \cite{WF16,T94,ZCJXS17}. Localization methods rely on the timely analysis of measurements such as angles of arrival (AOA), time differences of arrival (TDOA) between sensors, frequency differences of arrival (FDOA) between sensors or a combination of them ~\cite{Poisel12,MO7,GG3}. Therefore, an accurate synchronization of the time reference between sensors is essential. This can be achieved via satellite positioning systems or through packet-based protocols such as WhiteRabbit~\cite{CLWBS11}. However, both techniques are vulnerable to timing attacks, which constitutes a weakness for localization systems. Attackers are able either to spoof the signal~\cite{ZLQ12} or to insert a delay box on the links used for the synchronization communication~\cite{BALB16}. Such a delay box modifies asymmetrically the length of the communication paths, which indirectly injects delays in the time reference of sensors. The transmitted synchronization data is untouched by the attacker, hence it still satisfies all cryptographic security requirements in place such as authentication, integrity or confidentiality. A timing attack produces positive or negative offsets between the time reference of the sensors. The content of the data received and transmitted by the sensors also remains protected by the traditional cybersecurity protocols. The only noticeable effect of such an attack is if it affects the function of the system. Specifically, if it does not result in a misestimation of the location, then it is not detectable. Furthermore, an undetectable attack is required to be accepted and implemented by the clock controller of the sensors, a too large delay is flagged and raises suspicion. This is achieved by injecting small and gradually increasing delays. Although they require tampering with the communication network, timing attacks do not require any physical access to the potentially guarded sensors. Observe that it is realistic to assume non-guarded links between protected sensors. 

In this paper, we focus on the TDOA-based localization of a passive source from a network of fixed sensors whose time references could be maliciously manipulated. TDOA measurements offer high precision hence are widely used. However, they are easily attacked because they are particularly sensitive to timing errors. Due to the high propagation speed of the signal, a small synchronization error can lead to a large range difference error between the two sensors and the source. For example, if $3 \mu s$ are added to a TDOA measurement, then the corresponding range difference is increased by approximately $900$m. Consequently, an attacked network could become unable to localize sources or an attacked vehicle could unintendedly  enter on a wrong territory. Note that timing attacks are also a threat to other types of networks. For example, the control and operation of Smart Grids require an understanding of the system state at specific time intervals. It was shown that undetectable timing attacks on grid sensors are feasible and that they can lead to an incorrect state estimation, which in turn can result in a blackout or in asset degradation~\cite{SDBDBP19}. 

Our first contribution is to study the effect of timing attacks on the TDOA-based localization of an unknown source.
We show that the delays between sensors can lead to a misestimation of the source location. We inject a few microseconds into one sensor and obtain an estimate that is approximately $1$ km away from the true position of the source. We further explain how an attacker can compute positive or negative delays such that the localization process results in a specifically chosen misestimation. We also show that residual analysis does not enable the detection and identification of timing attacks.

Our second contribution is to propose a TDOA-localization technique that is robust against timing attacks. It works in two phases, the first analyses the error in TDOA measurements received from a known calibration source. As a result, it defines a weight for each pair of sensors; this reflects the confidence we have in their time synchronization. The second phase of our solution then uses the weighted least-squares (WLS) estimator with the newly created weights and the TDOA measurements received from the unknown source. Subsequently, our method either identifies the network as being too corrupt to localize, or gives a corrected estimate of the unknown position along with a confidence metric.

Our calibration technique requires the use of a trusted source of known coordinates. To our understanding, it is realistic to assume the existence of such a known source: a sensor of the localization network or a vehicle equipped with an emitter can be used for the calibration phase. In this first phase, we compare true TDOAs with observed TDOAs measured from signals emitted by the known source. We do not require the known source to be part of the synchronized network because our technique does not require the use of timestamps from the known source. In fact, our solution is immune to a timing attack on the known source. However, nothing prevents an attacker from storing emitted calibration signals in order to replay them in the direction of the attacked sensors at times and locations of his choice. For example, he could replay them in a manner that compensates for the introduced attack delays. To counter such an attack, we propose an encrypted authenticated challenge-response scheme where the calibration source is triggered to emit a one-time response signal. 

The rest of the paper is structured as follows. In Section~\ref{sec:related}, we discuss related works. We describe our system model in Section~\ref{sec:preliminaries}, together with technical background on TDOA-localization in an unattacked environment. We define the attacker's capabilities and study the effect of timing attacks in Section~\ref{sec:attack}. We present our calibration-based robust localization technique in Section~\ref{sec:defense}. We show how to counter replay attacks against our solution in Section~\ref{sec:replay}. We present the numerical results of the evaluation of the performance of our solution and of the confidence metric in Section~\ref{sec:perfeval}. Finally, we conclude the paper in Section~\ref{sec:ccl}.

\section{Related work}\label{sec:related}

In recent years, the subtleties of TDOA-based localization of passive sources sparked the interest of researchers. Due to the accuracy of this localization technique, its use is widespread. Its sensitivity to sensor location errors, oscillator-frequency-synchronization errors and time-synchronization errors between sensors has been the topic of various papers.
The authors of~\cite{FHJ3,FHJ32} studied the effect of phase and frequency-synchronization errors on the TDOA estimation, for different types of oscillators in the cases of single and multi-source localization. They propose a technique~\cite{FHJ4} to estimate both the TDOA measurement and the frequency error between sensors at low computational and memory complexities when the oscillator frequency error between two sensors is assumed to be non-zero and constant.
Similarly, the authors of~\cite{ZXH13,ZXH15} propose different techniques for estimating the TDOA between sensors, including oscillator phase and frequency errors. Their techniques are based on the Maximum Likelihood estimation of the TDOA, and one of them also estimates the frequency error of the oscillators. Then, similarly to this paper, the authors of~\cite{CWYW18,WYTCW18} focused on the localization of passive sources in systems of moving sensors that suffer from sensor position errors and from clock-synchronization bias between sensors. Their model, however, assumes that sensors are divided into groups within which sensors are time-synchronized, and that timing offsets are present only among different groups. In this paper, we assume that sensors are fixed at known locations and that they are all spaced out, therefore we assume that there can be time offsets between all sensor pairs. Furthermore, unlike in the previously mentioned papers, we assume that the synchronization offsets are not due only to the use of inaccurate hardware but also to the presence of malicious activity. As explained in Section~\ref{sec:attack}, we consider that an attacker is able to introduce time offsets in the clock of sensors in such a way that the resulting TDOA measurements seem plausible and intersect well, at a distant target location.

\section{System model}~\label{sec:preliminaries}

Consider a network of $N$ time-synchronized sensors $S_i$, $1 \leq i \leq N$ with known coordinates  $(x_i,y_i,z_i)$.
Suppose that a moving source $S$ of unknown coordinates $(x,y,z)$ produces a continuous signal $s(t)$. It is received by several network sensors  in the following form: $r_i(t)=s(t-\Delta_i)+e_i$, where $\Delta_i$ is the time needed for the signal to travel from the source to sensor $S_i$ and $e_i$ is Gaussian noise. The receiving sensors then simply timestamp the received signals and transmit them to a centralised control center. In order to compute an estimate $\widehat{\Delta}_{ij}$ of the true but unobservable TDOA $\Delta_{ij}$ between each pair of receiving sensors $(S_i,S_j)$, the control center then uses correlation techniques ~\cite{Piersol81,FHJ2,WK83} on the signal samples
\begin{equation}\label{model}
\widehat{\Delta}_{ij}=\Delta_{ij}+e_{ij}=\Delta_i-\Delta_j+e_{ij},
\end{equation}
where $e_{ij}$ is the noise associated with the estimated delay, in other words, the difference between the true and the estimated delay. Note that $e_{ij}$ is not equal to $e_{i}-e_{j}$.
Each TDOA $\Delta_{ij}$ defines a hyperbola on which the source should lie. The variance of the estimated delay defines a zone of probable location of the source along the corresponding hyperbola. By aggregating several measurements, the source is then estimated to be in a probable zone defined by the intersecting hyperbolae. The estimation of the source location is not trivial as it is a quadratic non-convex problem. The relation between an estimated delay $\widehat{\Delta}_{ij}$ and the source coordinates is given by the following equation \\
\begin{equation}
\widehat{\Delta}_{ij} =\frac{d(S_i,S)-d(S_j,S)}{c}  +  e_{ij} ,\\
\end{equation}
where $c$ is the propagation speed of the signal and \\ $d(S_i,S)=\sqrt{(x_i -x)^2+(y_i - y)^2+(z_i - z)^2}$ is the distance between sensor $S_i$ and the unknown source $S$. The unknowns in this equation are the source coordinates $(x,y,z)$ and the noise $e_{ij}$. This noise is distributed according to a centered normal distribution $\mathcal{N}(\mu_{ij}, \sigma_{ij})$ with $\mu_{ij}=0$ and where $\sigma_{ij}$ is unknown. \\ 
Supposing that the noise of the received signal at various sensors is i.i.d with same SNR, the covariance matrix is~\cite{CH94,DD09}
\begin{equation}\label{covariance}
K = \sigma^2
\begin{pmatrix}
1 & 1/2 & \cdots & 1/2 \\
1/2 & \ddots & \ddots & \vdots \\
\vdots & \ddots & \ddots & 1/2 \\
1/2 & \cdots & 1/2 & 1 \\
\end{pmatrix},
\end{equation}
where $\sigma^2$ is the TDOA noise variance. For simulations in the litterature, the standard deviation is often set arbitrarily to a plausible value such as $1.83 \mu s$ or $0.183 ns$. Nevertheless, more precise formulas to compute the standard deviation as a function of the SNRs of sensors are given in~\cite{Poisel12, Quazi}
\begin{itemize}
\item An aggregated SNR (not in dB) is computed from the SNRs of the two sensors $\gamma_{ij}=\text{SNR}(\gamma_i, \gamma_j ) $\\
\begin{equation}\label{eq:SNR}
\frac{1}{\gamma_{ij}}=\frac{1}{2}\left( \frac{1}{\gamma_i} + \frac{1}{\gamma_j}+\frac{1}{\gamma_i \gamma_j} \right) .
\end{equation}
\item For a low SNR value $\gamma_{ij}$, the standard deviation is given by
\begin{equation}\label{eq:lowsnr}
\sigma_{ij}= \sqrt{\frac{1}{8 \pi^2}}\frac{1}{\gamma_{ij}}\frac{1}{\sqrt{T_{int}W}}\frac{1}{f_0}\frac{1}{\sqrt{1+\frac{W^2}{12f_0}}}, 
\end{equation}
where $T_{int}$ is the integration time of the signal for one measurement, $W=f_2-f_1$ is the frequency bandwidth and $f_0$ is the center of frequency.
\item For a high SNR value $\gamma_{ij}$, the standard deviation is given by
\begin{equation}\label{eq:highsnr}
\sigma_{ij}= \sqrt{\frac{3}{4 \pi^2 T_{int}}}\frac{1}{\sqrt{\gamma_{ij}}}\frac{1}{\sqrt{f_2^3-f_1^3}} .
\end{equation}
\end{itemize}

From noisy TDOA measurements, both geometrical and analytical techniques of localization can be found in the litterature~\cite{GG3,DD09}. The non-linear least- squares (LS) estimator is a widespread localization technique that takes the noisy TDOA measurements as input and searches for a solution $(x,y,z)$ minimizing the sum of squared errors $$\arg \min_{x,y,z} \sum_{i>j} \left( \frac{d(S_i,S)-d(S_j,S)}{c}  - \widehat{\Delta_{ij}}\right)^2.$$ This estimator can be modified to solve the weighted least-squares problem (WLS) using the covariance matrix of the TDOA measurements. Throughout the rest of the paper, the Levenberg-Marquadt (LM) algorithm is used to solve this non-convex optimization problem, thus estimating the coordinates of unknown sources. More discussion on the LM algorithm is provided in Section~\ref{sec:perfeval}.

In order to reduce the complexity of storage and of the estimation process, a widespread technique is to consider only linearly independent measurements by considering only the TDOA measurements with respect to a reference sensor. This reduces the number of equations from $ \binom{N}{2}$ to $N-1$, where $N$ is the number of available sensors. However, this technique induces a loss of redundancy which can be fatal to the localization system in the case of an attack.

\section{Impact of time-synchronization attacks}\label{sec:attack}
In this section, we show how an attack on the time reference of one or more sensors alters the localization process presented in Section~\ref{sec:preliminaries}. We begin by defining the capabilities of the attacker.
\subsection{Attack Model}\label{sec:model}
We consider two timing attack models: one of them is referred to as the weak attack model, and the other is referred to as the strong attack model.
In both cases, we suppose that an attacker is able to introduce an offset $a_i\in \mathbb{R}$, which can be positive or negative, to the time reference of sensor $S_i$. Recall that such an attack does not require physical access to the sensors as it can be achieved via signal spoofing or delay-box insertion, depending on the preferred synchronization technique. With this capability, the goal of the attacker is to introduce errors in TDOA measurements, thus provoking a misestimation of the location of an unknown source. In the weak attack model, the attacker is not able to choose the misestimation, his goal is to create errors in the localization. In contrast, in the strong attack model, we further suppose that the attacker knows the true source coordinates and the network topology, namely the sensor coordinates. In this case, the objective of the attacker is to ensure that the localization process results in a specific targeted misestimation. 

\subsection{Impact on Localization}
Introducing delays $a_i$ and $a_j$ to sensors $S_i$ and $S_j$ respectively, adds components to Eq.(\ref{model}) 
\begin{equation}\label{eqatt}
\widehat{\Delta}_{ij}=\Delta_{ij}+a_i-a_j+e_{ij}=\Delta_{ij}+\mu_{ij}+e_{ij},
\end{equation}
where $\mu_{ij}=a_i-a_j$ is the introduced delay difference between the two sensors. Observe that $\widehat{\Delta}_{ij}-\Delta_{ij}$ is still distributed according to $\mathcal{N}(\mu_{ij}, \sigma_{ij})$, with $\mu_{ij}=0$ if no delays are inserted and $\mu_{ij}=a_i-a_j$ otherwise. Therefore, introducing delays does not guarantee an impact on measurements and thus on the localization. In fact, the attack is meaningful only if there is a non-negligible delay difference between the time references of at least two sensors. If an attacker introduces the same delay to all sensors of the network, they will all be under attack yet remain synchronized with each other. Hence, the functionality of the system will not be altered and the presence of malicious activity will be undetected. When the attack is such that the difference $\mu_{ij}> \sigma_{ij}$ is non-negligible in comparison to the Gaussian noise, the measurement $\widehat{\Delta}_{ij}$ and its corresponding hyperbola are significantly modified. As a result, the localization process fails to give an accurate estimate. Next, we give two attack scenarios that illustrate the actions of an attacker as in the two proposed attack models of Section~\ref{sec:model}.

In both scenarios, we consider a two-dimensional grid of side of $20$ km with four sensors placed as in Figures~\ref{fig:impactnano} and~\ref{fig:impacttarget}. The unknown source emits a signal that propagates to the four sensors. We suppose that the sensors send their received signal samples to a control center that processes them. The resulting TDOA measurements are given as input to the WLS estimator as described in Section~\ref{sec:preliminaries}. In this simulation, we set the noise standard deviation to $\sigma=2.192ns$ for all TDOA measurements, this value is further discussed in Section~\ref{sec:perfeval}. In the the weak attack-model scenario, the attacker delays the time reference of sensor $S_1$ by $2.47 \mu s$, which modifies the three corresponding hyperbolae, drawn in red in Figure~\ref{fig:impactnano}. The resulting WLS estimate of the unknown source location is incorrect by approximately $1$ km, thus illustrating that localization from TDOA measurements is highly sensitive to time offsets. As the injected delay increases, the accuracy of the estimate decreases. Note that in this scenario, the full set of measurements was considered. Supposing that we considered only three linearly independent measurements with respect to sensor $S_1$, then the estimate would be even less accurate as the WLS estimator would have received only wrong measurements as input. In other words, only the red hyperbolae from Figure~\ref{fig:impactnano} would be taken into account and none of the black ones. 

\begin{figure}
	\centering
	\includegraphics[width=0.8\columnwidth]{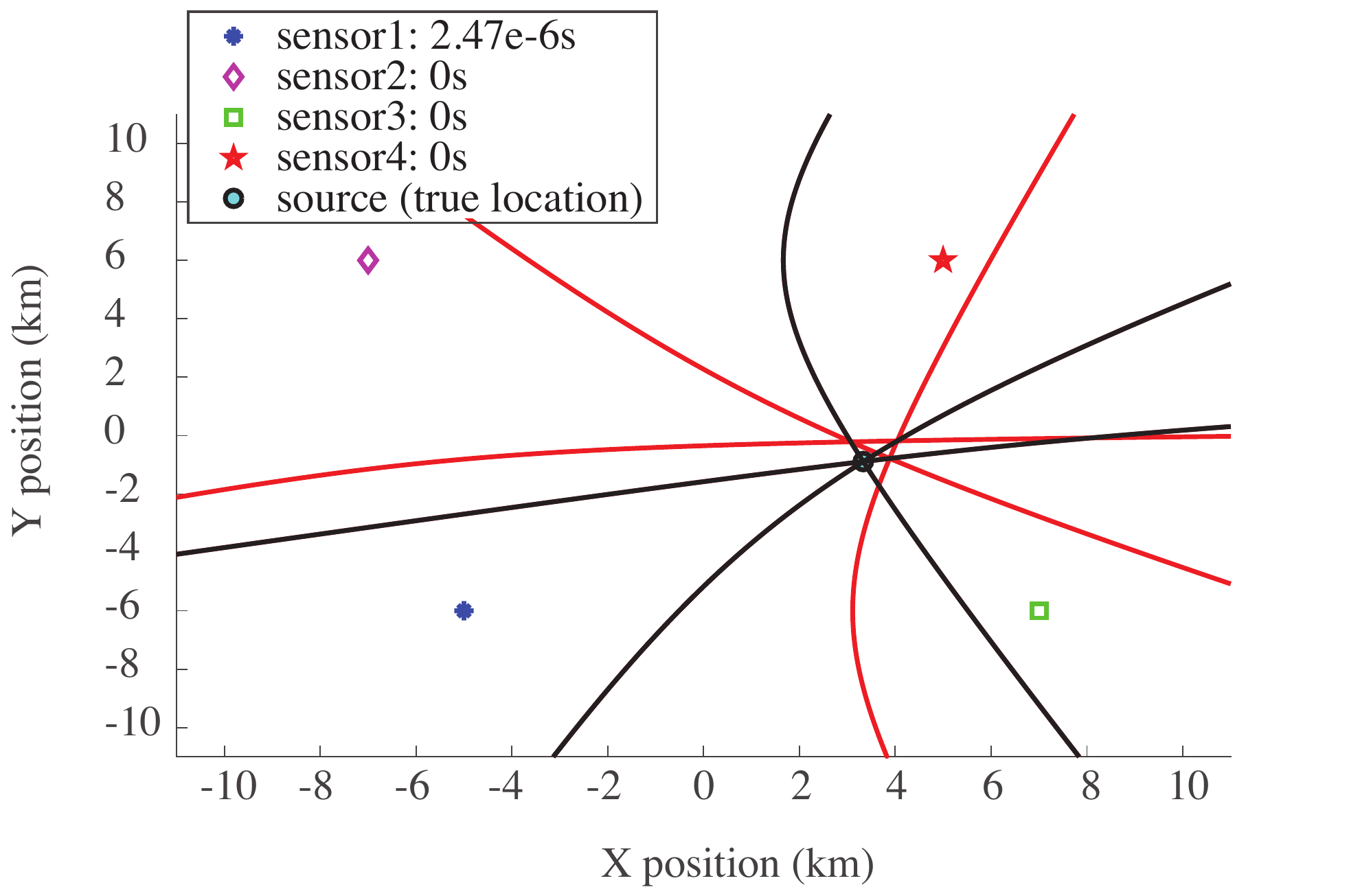}
	\caption{Weak-attack model scenario: sensor $1$ is delayed by $2.47\mu s$. The red TDOA hyperbolae are shifted by this delay and the resulting source estimate is incorrect by approximately $1$ km.}
	\label{fig:impactnano}\vspace{-10pt}
\end{figure}

In the strong attack-model scenario, the attacker knows the coordinates of the sensors and the true coordinates of the source. With such knowledge, he is able to compute the delays to be injected such that the estimation process results in a specific targeted misestimation: 
\begin{itemize}
\item From the source and sensor coordinates, he computes the true delays of propagation of the signal between the source and the sensors $\Delta_i=\frac{d(S_i,S)}{c},~\forall 1 \leq i \leq N$.
\item Similarly, he computes the true delays of propagation of the signal between the targeted misestimation location and the sensors $\Delta_i^t,~\forall 1 \leq i \leq N$.
\item The attack is simply the difference between the two: the delay to inject to sensor $S_i$ is $a_i=\Delta_i^t - \Delta_i$.
\end{itemize}
This attack is illustrated in Figure~\ref{fig:impacttarget}, where the delays are computed specifically such that all hyperbolae are modified in a plausible manner, intersecting near the targeted misestimation location. The resulting estimate is, as chosen by the attacker, almost $9$ km away from the true source location. 


\begin{figure}
	\centering
	\includegraphics[width=0.8\columnwidth]{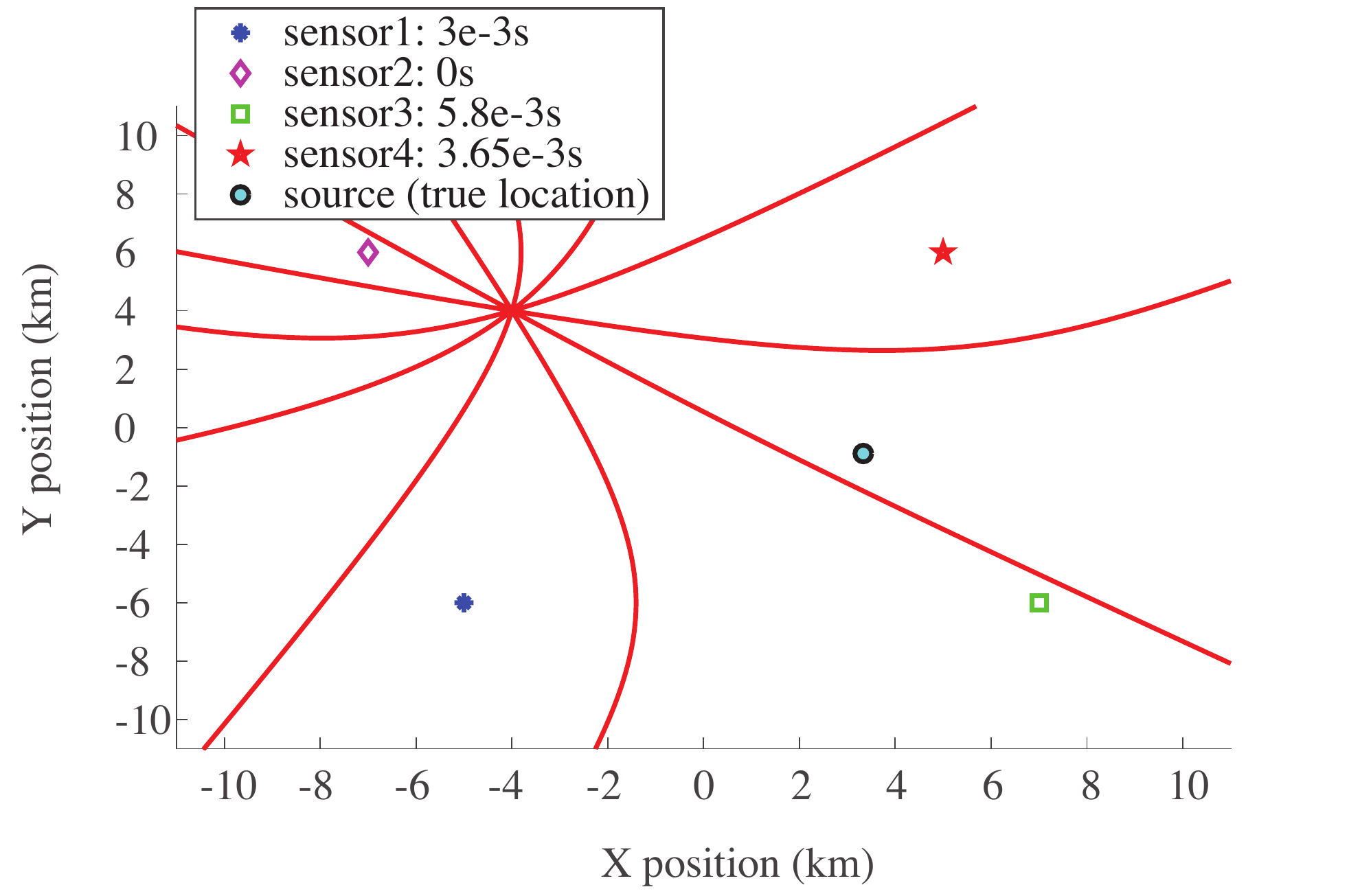}
	\caption{Strong-attack model scenario: delays are strategically computed which results in a specific misestimation almost $9$ km away from the true source location.}
	\label{fig:impacttarget}\vspace{-10pt}
\end{figure}

\subsection{Residual Analysis}~\label{subsec:residual}

Once an estimate $S_{est}$ is computed from measurement values, it is useful to compute and analyze the residuals in order either to assess the accuracy of the estimator or to attempt to detect and identify bad data in the measurements. For a pair of sensors $(S_i,S_j)$, the corresponding residual is computed as follows $$\frac{d(S_i,S_{est})-d(S_j,S_{est})}{c}  - \widehat{\Delta_{ij}} .$$ The residuals give insight on how well the estimate fits the measurements. Hence, if all hyperbolae intersect near the estimated location, then the residuals will have small values. However, if the points of intersection of the hyperbolae define a large probable zone of location, then the resulting estimate will be far from some or all hyperbolae, thus resulting in large residuals. Residual values are given in Table~\ref{residuals} for the no-attack scenario and for the scenarios depicted in Figures~\ref{fig:impactnano} and~\ref{fig:impacttarget} that correspond to the two attack models. Observe that the strong attack-model scenario and the no-attack scenario have similar residuals of order of magnitude below a nanosecond. This is as expected because, in both cases, the estimate satisfies well all measurements. Against the strong attack model the residual analysis fails to detect the presence of malicious activity. In the weak attack-model scenario, Table~\ref{residuals} shows that three residuals have values by three orders of magnitude larger than their corresponding values in the no-attack scenario. However, these large residuals are not all related to sensor $S_1$. In this case, the residual analysis succeeds in detecting the presence of malicious activity in the system but fails to identify it clearly. Overall, residual analysis is misleading for this study as it either fails to detect the presence of an attack or fails to identify untrustworthy measurements. Another approach for building resilience against timing attacks is proposed in the following section.

\begin{table*}[h!]
\centering
\begin{tabular}{|c|c|c|c|c|c|c|}
\hline
attack type 																		& $\Delta_{12}$ (s) & $\Delta_{13}$ (s) &$\Delta_{14}$ (s) &$\Delta_{23}$ (s) &$\Delta_{24}$ (s) &$\Delta_{34}$ (s) \\ \hline
weak attack model scenario					& $8.11e-11$  & $1.59e-9$     & $-3.56e-6$ & $-2.33e-9$ & $-3.57e-6$ & $-3.57e-6$ \\ \hline
strong attack model scenario & $-3.25e-9$ & $-1.34e-9$  &  $-5.38e-10$&$-4.32e-10$ &$2.19e-9$ & $2.2e-9$\\  \hline
no-attack 																& $-1.68e-9$ & $-4.73e-10$ & $2.03e-10$ & $-4.2e-10$ & $-6.3e-11$  & $5.68e-10$ \\ \hline
\end{tabular}
\caption{ \label{residuals} TDOA residuals. The values are similar in the no-attack and strong attack-model scenarios: the strong attack is undetected by residual analysis. Some values are increased by three orders of magnitude in the weak attack-model scenario, a majority vote points to $S_4$ as being attacked but the only attacked sensor is $S_1$: residual analysis detects malicious activity but does not identify it clearly.}
\end{table*}

\section{Calibration-based robust localization}\label{sec:defense}
As mentioned above, the analysis of residuals during the localization of an unknown source is not sufficient to counter timing attacks. In this section, we present a robust localization strategy that works in two phases. The first is a calibration phase that makes use of a known source to estimate the pairwise synchronized sensors of the network. The second phase of our strategy consists in the localization of an unknown source, given the results of the calibration process. We also show how to compute a confidence metric that gives insight on the accuracy of the estimated location. 

\subsection{Calibration Phase}\label{cal}

In this first phase, we use authenticated received signals emitted by known sources of known coordinates. Our technique then compares the resulting TDOA measurements with the true delays that should be observed. These true values are easily computed from the known coordinates as shown in the previous section for the computation of specific attack delays. The aim of the calibration phase is to define weights $w_{ij}$ for each pair of sensors $(S_i,S_j)$, reflecting the confidence level of their time-synchronization. Recall that $\widehat{\Delta}_{ij}-\Delta_{ij}$ is distributed according to $\mathcal{N}(\mu_{ij},\sigma_{ij})$ with $\mu_{ij}=0$ if $S_i$ and $S_j$ are time-synchronized. Therefore, the weight $w_{ij}$ must reflect the confidence with which we could declare that $\mu_{ij}=0$ given the true delay $\Delta_{ij}$ and multiple samples of $\widehat{\Delta}_{ij}$, denoted $\widehat{\Delta}_{ij}^1$,...$\widehat{\Delta}_{ij}^n$. For example, the weights could have binary values in order to define hard clusters within which sensors are time-synchronized: 
\[
w_{ij} =
\begin{cases}
& 0 \text{ if } \mu_{ij} > \sigma_{ij} \\
&1 \text{ if } \mu_{ij} \leq \sigma_{ij} \\
\end{cases} .
\]
Nevertheless, given noisy observations, the true cluster can only be estimated with a certain level of confidence. In practice, hard clustering methods are not able to minimize both the probabilities of false positives and false negatives. Therefore, we use a soft clustering method, i.e., we allow non-binary values $w_{ij}\in [0,1]$.

Hypothesis testing and the z-test in particular, are often used in order to determine whether a sample data-set is from a population with a specific mean. The z-test can be used only if the sample data is assumed to follow a normal distribution of known standard deviation as it is the case for the $n$ samples of $\widehat{\Delta}_{ij}-\Delta_{ij}$. The test computes the standardized statistic $$z_{ij}=\frac{\overline{(\widehat{\Delta}_{ij}-\Delta_{ij})}}{\frac{\sigma_{ij}}{\sqrt{n}}},$$ where $\overline{(\widehat{\Delta}_{ij}-\Delta_{ij})}$ corresponds to the sample mean of the $n$ observed delay differences. If it is truly the case that $\mu_{ij}=0$, then this standardized statistic $z_{ij}$ must be distributed according to $\mathcal{N}(0,1)$. Depending on the value of $z_{ij}$ and a predefined threshold, the test either accepts or rejects the hypothesis that $\mu_{ij}=0$. However, as mentioned above, the weights that we define are not constrained to have binary values. We define them as a function of the z-test p-values.Specifically, the weight $w_{ij}$ is a function of the probability of observing a test statistic larger or equal to $z_{ij}$ given that $\mu_{ij}=0$. This probability is computed as $\text{erfc}(z_{ij}/\sqrt{2})$, where $\text{erfc}(x)=\frac{2}{\sqrt{\pi}}\int_x^\infty e^{-t^2}dt$ is the well-known complementary error-function. As the difference between the measured and the true delay decreases, the corresponding p-value increases. In order to amplify the weight differences between pairs of sensors with reasonably large and extremely small p-values, we define weights to be $w_{ij}=(pvalue_{ij})^{1/v}$, where $v \in \mathbb{R}$. The exponent $v$ can be optimally chosen to maximize the weight difference for two specific p-values: $v=15.0776$ maximizes the weight difference for p-values $10^{-4}$ and $10^{-10}$. Nevertheless, our simulations show that all choices of $v \in [10,30]$ give satisfactory localization results with negligible variance. The computed weights are used in the localization process of any unknown source until they are updated. This means that any TDOA measurement that results from the correlation of signals received by sensors $S_i$ and $S_j$ will be weighed by $w_{ij}$. Note that if all p-values are equal to zero, no sensor pair data will be trusted to be used in the localization process. If all p-values are very low but larger than zero, then all sensor data is taken into consideration. But our confidence in the synchronization of all sensor pairs is low, hence we expect low accuracy. Whereas if all p-values are high, the confidence in the synchronization is high throughout the network, and we expect to obtain accurate estimates.

In order to give insight about the level of accuracy with which the localization process is able to compute an estimate of the location of any unknown source, we propose to add the computation of a confidence metric to the calibration phase. As in two dimensions, the minimal number of measurements required to localize a source is two, the accuracy of a location estimate depends on how well the second most trustworthy sensor pair seems to be synchronized; this is captured by the second best p-value. Furthermore, the accuracy improves with redundancy, hence if the third best p-value is also high, we expect that the estimate will be even more accurate. Hence, our proposed confidence metric \textit{cfd} is defined as the sum of the second and third best p-values to the power $1/v$, divided by two. Similarly in three dimensions, one level of redundancy is achieved by including the fourth best p-value and by dividing the sum by three instead of two. More discussion on this metric is provided in Section~\ref{sec:perfeval}. The operations of the calibration phase are recapitulated in Algorithm~\ref{alg:cal}. It shows how to compute the weights for each sensor pair and the confidence metric of the network at a given time.

\begin{algorithm}[h]
  \caption{Define-weights($\mathcal{N}, S_c, \sigma, c, \widehat{\Delta^1},...,\widehat{\Delta^n}, v,D$)}
  \begin{algorithmic}
  	\Require $\mathcal{N}$ (network of sensors $S_i$, $1 \leq i \leq N$), $S_c$ (known calibration source), $\sigma$ (standard deviation of TDOA measurements), $c$ (signal speed), $\widehat{\Delta^1},...,\widehat{\Delta^n}$ ($n$ symmetric matrices of TDOA measurements from $S_c$), $v$ (weight function exponent), $D$ (dimension 2D or 3D)
  	\For {$S_i \in \mathcal{N}$}
  	  \State $\Delta_i \leftarrow \frac{d(S_i,S_c)}{c}$
  	\EndFor 
	\State $weights \leftarrow \emptyset$
	\State $pvals \leftarrow \emptyset$
	\For {$(S_i, S_j)\in \mathcal{N}^2, i \neq j$}  	  	
  	
  		\State $(e_{ij}^1,...,e_{ij}^n) \leftarrow (\widehat{\Delta^1}_{ij}-\Delta_i+\Delta_j,...,\widehat{\Delta^n}_{ij}-\Delta_i+\Delta_j) $
  		\State $pvalue \leftarrow \text{z-test}(e_{ij}^1,...,e_{ij}^n, \sigma_{ij})$
  		\State $pvals \leftarrow pvals \cup pvalue $
  		\State $weights \leftarrow weights \cup (pvalue)^{1/v} $
  		\EndFor
  		\If {$D=2$}
  			\State \textit{cfd} $\leftarrow \frac{(\text{max}_{2^{nd}}(pvals))^{1/v} + (\text{max}_{3^{rd}}(pvals))^{1/v} }{2}$
  		\Else
  			\State \textit{cfd} $\leftarrow \frac{(\text{max}_{2^{nd}}(pvals))^{1/v} + (\text{max}_{3^{rd}}(pvals))^{1/v} + (\text{max}_{4^{th}}(pvals))^{1/v}}{3}$  		\EndIf
  		\State $weights \leftarrow \frac{weights}{\text{sum}(weights)} $
    \Ensure $weights$, $cfd$
  \end{algorithmic}
   \label{alg:cal}
\end{algorithm}

\subsection{Robust-Localization Phase}

The purpose of the second phase of our technique is to sporadically localize unknown sources by using data from sensors that are possibly suffering from a timing attack. In Section~\ref{sec:preliminaries}, the WLS estimator was introduced with weights defined by the covariance matrix of the noise of the measurements. Our robust localization technique further weights the squared errors with the weights computed during the calibration phase. In other words, our solution is to search for $(x,y,z)$ minimizing $$ \sum_{i>j} \frac{w_{ij}}{\sigma_{ij}^2} \left( \frac{d(S_i,S)-d(S_j,S)}{c}  - \widehat{\Delta_{ij}}\right)^2.$$ Algorithm~\ref{alg:localization} describes how the robust localization phase works for a two-dimensional grid, when at least two sensor pairs have non-zero weights. It uses the function noise$\_$std$(\gamma_i,\gamma_j)$ to compute the noise standard deviation of TDOA measurement $\widehat{\Delta_{ij}}$ from the SNR values at sensors $S_i$ and $S_j$ according to equations~\ref{eq:SNR},~\ref{eq:lowsnr} and~\ref{eq:highsnr}. In the case where all weights are set to zero, the algorithm states that the system is too corrupt to reliably estimate the location of the unknown source.

\begin{algorithm}[h]
  \caption{Robust-localization($weights, \mathcal{N}, \widehat{\Delta},\gamma, D$)}
  \begin{algorithmic}
  	\Require $weights$ (computed by Algorithm~\ref{alg:cal}), $\mathcal{N}$ (network of sensors), $\widehat{\Delta}$ (matrix of received TDOAs from unknown source), $\gamma$ (vector of SNR values for each sensor), $D$ (dimension 2D or 3D)
  	\If {$| \text{nonzero}(weights) | \geq D$ }
  	\State $\sigma \leftarrow \emptyset$
  	\For {$(S_i, S_j)\in \mathcal{N}^2, i \neq j$}  	
  	\State $\sigma \leftarrow \sigma \cup \text{noise}\_\text{std}(\gamma_i, \gamma_j)$
  	\EndFor
  	\State estimate $\leftarrow$ WLS$(\frac{weights}{\sigma^2}, \mathcal{M}, \mathcal{N})$
  	\Else
  	\State estimate $\leftarrow \text{"corrupt$\_$system"}$
  	\EndIf  
    \Ensure estimate
  \end{algorithmic}
   \label{alg:localization}
\end{algorithm}

Recall that in the weak attack-model scenario presented in Section~\ref{sec:attack}, only sensor $S_1$ is attacked with a delay of $2.47 \mu s$. When no defense strategy is in place, the LM algorithm on the full set of measurements gives a WLS estimate approximately $1$ km away from the true source location. Using our robust localization technique with the LM algorithm, we obtain an estimate only $40$cm away from the true source location. Furthermore, recall that in the strong attack-model scenario considered in Section~\ref{sec:attack}, the attack delays were computed specifically such that the localization of a particular unknown source would result in a targeted location. The computed delays shown on Figure~\ref{fig:impacttarget} are: $a_1=3ms$, $a_2=0s$, $a_3=5.8ms$ and $a_4=3.65ms$. The obtained estimate was, as chosen by the attacker, approximately $9$ km away from the source. Whereas our technique flags all sensor pairs as not synchronized, and all weights are set to zero. As a result, our algorithm states that the system is too corrupt to give an estimate.

\section{Countermeasures against replay attacks}\label{sec:replay}

The calibration phase of our solution presented in Section~\ref{cal}, relies on signals received by the sensors and emitted by a calibration source of known coordinates. In this section, we suppose that an attacker seeks to perform an attack on the calibration phase of our solution in order to make it attribute wrong weights. In case of such a successful attack, our robust localization would discard trustworthy measurements and/or trust attacked measurements. As a result, our technique would fail to detect malicious activity and would discard correct measurements.

In this section, we consider two additional attack models in which the attacker targets the calibration phase of our robust solution. His goal is to provoke a wrong weight attribution, thus maintaining the undetectability of his ongoing timing attack or neutralising the localization system. In the weak calibration attack model, we suppose that the attacker is able to record signals from the calibration source in order to replay them at times and locations of his choice. For example, he could replay them in a manner that compensates for the introduced attack delays. In the strong calibration attack model, we further suppose that the attacker is able to jam signals emitted by the calibration source~\cite{AB15}. We assume that the attacker does not jam continuously but selectively. Note that a sensor being continuously jammed would be flagged as suspicious due to other identification methods, such as SNR analysis. The fact that an attacker with complete control on the flow of signals in the network would be all powerful, further justifies the assumption of a selective jamming. 

In order to prevent such attacks, we propose an encrypted authenticated challenge-response scheme between the calibration source (CS) and the control center (CC). For the duration of this protocol, we assume that CS is stationary and emits signals continuously and that CC is responsible for triggering CS into embedding specific responses in the emitted signal. The first iteration of the scheme is depicted in Figure~\ref{fig:replay}.
\begin{figure}
	\centering
	\includegraphics[width=\columnwidth]{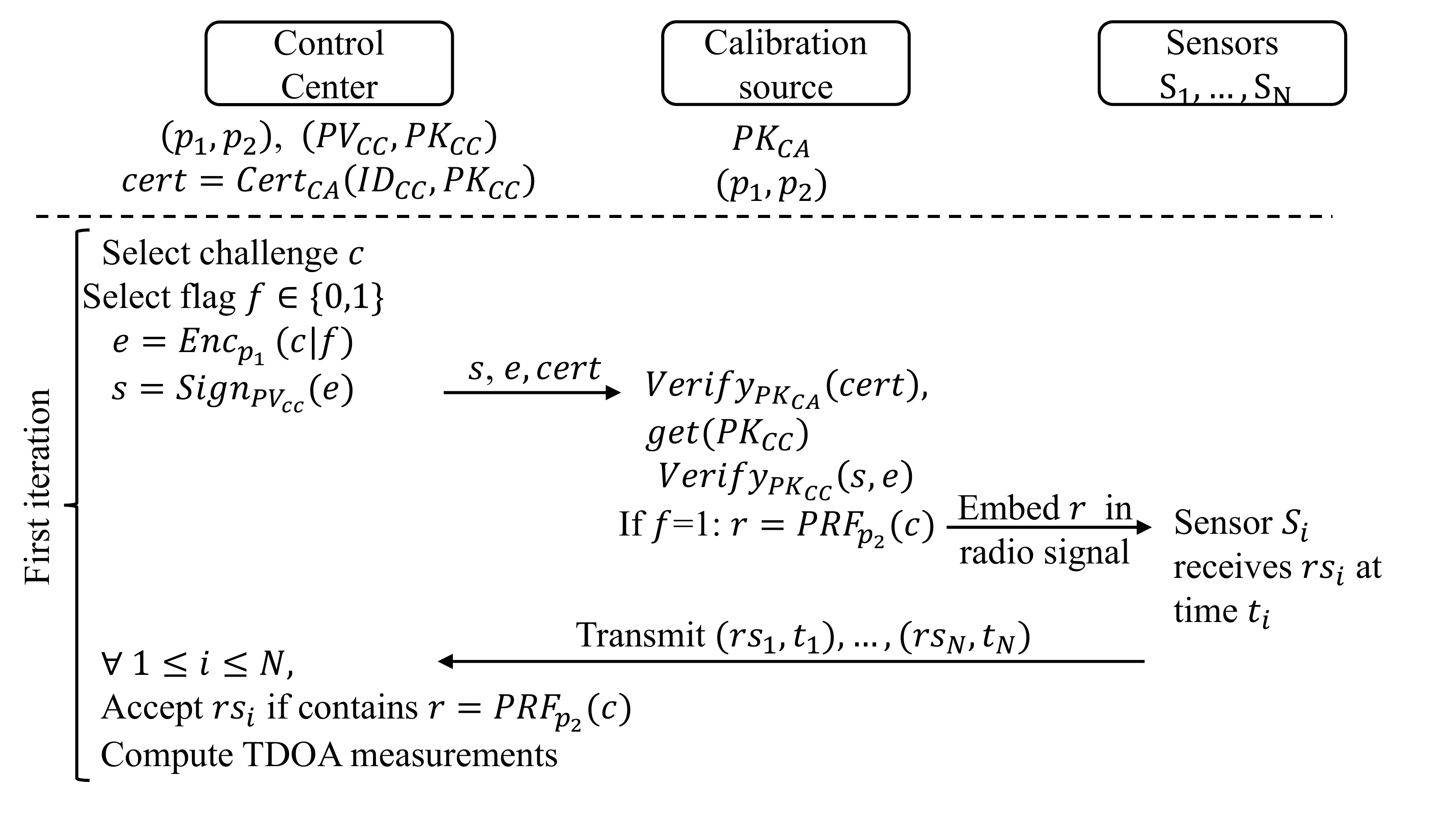}
	\caption{First iteration of the encrypted authenticated challenge-response scheme that counters replay attacks against the calibration phase of our robust-localization technique; following iterations are identical but don't include the $cert$ verification and the extraction of $PK_{CC}$.}
	\label{fig:replay}\vspace{-10pt}
\end{figure}
We suppose that due to a key infrastructure, CC has a certificate binding its identity with its public key $PK_{CC}$. We further suppose that CS knows the public key of the certificate authority $PK_{CA}$ and that CC and CS share two secrets $p_1$ and $p_2$ of large entropy. At the first iteration of the scheme, CC sends its certificate to CS who first verifies it using $PK_{CA}$ and then extracts the public key $PK_{CC}$. The latter will be used to authenticate CC to CS, at all subsequent iterations of the scheme. 

During the calibration phase, CC continuously sends the encryption of a one-time random challenge $c$ concatenated with a flag $f \in \{ 0,1\}$ that indicates whether or not CS should compute and embed a response $r$ in the emitted signal. The use of this flag enables us to hide from the attacker the locations in the signal where there are embedded responses. These responses are used to compute the TDOA measurements that are analyzed during the calibration phase. Note that, in order to be successful, the attacker needs to succeed in delaying a number                                                                                                                                                                                                                                                                                                                                                                                        of these responses embedded in the signal. Using the flag, instead of sending useful responses continuously, decreases the probability that a useful response is delayed by an adversary. This is due to the fact that in order to be undetected, the adversary is required to jam only selectively. Therefore, if we assume that there is a maximal frequency with which he can jam signals while remaining undetected, if he does not know which parts of the signal are useful and should be jammed, then the proportion of TDOA measurements that he successfully attacks is decreased. Note that, for this to be true, it is required that the flag be picked at random and that the signal emitted by CS without an embedded response be indistinguishable from when it contains a response.

The one-time challenge $c$ concatenated with $f$, is encrypted into ciphertext $e$ via a symmetric encryption scheme using secret key $p_1$. In order to authenticate itself to CS, CC also sends a signature of $e$ computed using its private key $PV_{CC}$. Upon reception of $(e,s)$, CS verifies the validity of the pair sent by CC. If it is valid, CS decrypts $e$ using $p_1$ and extracts $f$. If the flag is equal to $1$, then it computes a response $r$ that corresponds to the result of a chosen pseudo-random function with $c$ as input and secret $p_2$ as the key. Then, CS embeds $r$ in the signal. In contrast, if the flag is equal to $0$, CS does not embed anything in the signal. This signal is received by various known nearby sensors who simply timestamp everything before transmitting to CC. Finally, CC accepts only signals containing the valid response $r$ that are transmitted by specific nearby sensors able to correctly decode signals from CS. It then correlates the accepted signals to obtain TDOA measurements. Note that to analyze the correlation of the signals, the signals are aligned according to the timestamp that was associated by the receiving sensors. Hence, introducing delays in the transmission of received signals between the sensors and CC does not affect the resulting TDOA measurements. This constitutes one iteration of the overall scheme. We denote by $m$, the number of iterations with $f$ set to $1$. Hence, at the end of the process, there are $m$ resulting TDOA measurements per sensor pair.  Recall that the calibration phase described in Algorithm~\ref{alg:cal} requires $n$ observed measurements for each sensor pair. Below, we discuss the selection of $n$ among $m$ measurements per sensor pair and analyze the security of our scheme firstly against a weak calibration attacker and secondly against a strong calibration attacker. 

First, in order to enforce security against a weak calibration attacker, we propose to set $m=n$ and to discard signals containing a reoccurring valid specific response $r$. We show in Theorem~\ref{th:weak} that this scheme is secure against a weak calibration attacker with overwhelming probability, i.e., the attacker is unable to inject delays in the TDOA measurements of the calibration phase. The resulting TDOA measurements are then given as input to Algorithm~\ref{alg:cal} as described in Section~\ref{sec:defense}. 

\begin{theorem}\label{th:weak}
Assuming that challenges are unique, that $p_1$ and $p_2$  are secret, that signals with a reoccurring $r$ are discarded, and that all relay attacks take more time than the direct signal to propagate to the sensors, then the encrypted authenticated challenge-response scheme with $m=n$ iterations is secure in the weak calibration attack model with overwhelming probability.
\end{theorem}

\begin{proof}
Since $p_1$ and $p_2$ are secrets of large entropy and since $c$ is unique and also a secret of large entropy, there is a negligible probability that an adversary manages to forge a valid response $r$ that corresponds to the latest challenge $c$. Therefore, we assume that he can emit signals containing valid responses only by replaying those emitted by CS. As we assume that all relay attacks take more time than the direct signal takes to propagate to the sensors, all replayed signals will be received by CC with timestamps that are more recent than the timestamps of the direct signals. Hence, as they were already received by CC, replayed signals are all discarded. We conclude that the encrypted authenticated challenge-response scheme with $m=n$ iterations is secure in the weak calibration attack model, as long as the adversary is unable to forge an $r$, hence it is secure with overwhelming probability.
\end{proof}

Second, when we consider a strong calibration attacker able to jam signals, it is possible that the first occurring response $r$ is a replayed signal. In this scenario, we assume that the attacker can selectively jam signals emitted by CS and replay them such that the direct signal is never received by the sensors.
In order to enforce security in this strong attack model, we propose to set $m$ much larger than $n$. In other words, we propose to select a small portion $n$ of the received measurements $m$ to use as input for Algorithm~\ref{alg:cal}. As mentioned earlier, we suppose that the attacker does not jam continuously but selectively so that only $qm$ measurements are successfully delayed by the adversary, where $q$ is the proportion of attacked measurements. As the TDOA noise is distributed according to a Gaussian distribution, the $(1-q)m$ unattacked measurements are distributed according to a Gaussian distribution centered in the observable TDOA value that depends on the coordinates of the sensors, the coordinates of CS and the possible time-synchronization offset between the sensors. In contrast, the $qm$ attacked measurements are expected to be distributed according to another Gaussian distribution centered around a value that depends on the delay difference with which the signal is replayed to the different sensors. 

Our strategy is to analyze the distribution of the $m$ observed measurements in order to extract the center of the tallest Gaussian distribution and to select the $n$ measurements that are nearest to it. More specifically, we use a binning algorithm on all the received measurements; the algorithm returns $b$ bins of uniform width covering the range of the $m$ measurements. This reveals the shape of the sample distribution. We then extract the bin of highest density and iteratively extract its surrounding bins, until the sub-sampled dataset is of cardinality at least $n$. Then, we estimate the probability density function of the selected data and search for its peak value. As a result, we obtain an estimate of the center of the tallest underlying Gaussian distribution, in other words, of the TDOA value that should be observed. Finally, we select the $n$ measurements that are nearest to the newly found estimate. Observe that this technique works in favour of the system only when there are fewer attacked than unattacked measurements during the calibration phase. Otherwise, the tallest Gaussian distribution would correspond to the underlying distribution of the attacked measurements. These selected measurements can then be given as input to Algorithm~\ref{alg:cal} as described in Section~\ref{sec:defense}. We show the following claim through numerical analysis in Appendix.
\begin{claim}\label{strong}
As long as the proportion $q$ of attacked measurements is lower or equal to $0.45$, our strategy is secure against a strong calibration attacker.
\end{claim}

\section{Performance Evaluation}\label{sec:perfeval}
In this section, we evaluate the performance of our solution presented in Section~\ref{sec:defense} by considering various attack scenarios. We show through Matlab simulations that our solution is robust to timing attacks and that the confidence metric is reliable. We start by defining the testing environment in two dimensions. We then consider a three-dimensional simulation.

\subsection{Two-Dimensional Testing Environment}
We consider a two-dimensional grid of side of $20$ km on which we place four sensors, as illustrated in figures~\ref{fig:impactnano} and~\ref{fig:impacttarget}. We assume that at the time instant of the analysis, the unknown source is located at coordinates  $[3333.3,-889.1111]$ and the known calibration source at $[0,-4000]$ all in meters, with respect to the center of the grid. In Section~\ref{sec:3D}, we show a three-dimensional simulation on the same grid, to which we add altitude coordinates.
In our simulations, we assume that the TDOA noise is i.i.d with a standard deviation of $2.192n$s for all measurements. We computed this standard deviation from Eqs.~\ref{eq:SNR},~\ref{eq:lowsnr},~\ref{eq:highsnr} with the following parameters:
\begin{itemize}
\item the integration window of the signal $T_{int}=60$ ms,
\item the bandwidth of the signal $W=1$ MHz,
\item the center of frequency $f_0=30'000$ Hz,
\item the SNR $\gamma_i=3$ dB $\forall S_i$.
\end{itemize}
For the calibration phase, we created simulated TDOA measurements $\widehat{\Delta}_{ij}$ from the \emph{known} calibration source for all sensor pairs $(S_i,S_j)$ with the following procedure:
\begin{itemize}
\item we use the coordinates to compute the true time $\Delta_i$ taken by the signal to propagate from the calibration source to sensor $S_i$ for both sensors,
\item we compute the true TDOA $\Delta_{ij}=\Delta_i - \Delta_j$,
\item depending on the attack scenario, we add delays \\ $\Delta_{ij}'=\Delta_{ij} + a_i - a_j$,
\item we add Gaussian noise $\widehat{\Delta_{ij}}=\Delta_{ij}'+e_{ij}$ \\ with $e_{ij} \in \mathcal{N}(0,2.192 e-9)$,
\item we repeat this $n=15$ times in order to have $15$ measurements for each pair of sensors.
\end{itemize}
Our calibration procedure defined by Algorithm~\ref{alg:cal} then uses the simulated measurements in the following way:
\begin{enumerate}
\item We compute the observed error $\widehat{e_{ij}}=\widehat{\Delta}_{ij}-\Delta_{ij}$ for each measurement, \emph{using the known coordinates} of the calibration source.
\item For each pair of sensors, we compute the p-value resulting from the z-test with the $15$ observed errors $\widehat{e_{ij}}$.
\item We compute the weights $w$ by exponentiating all p-values to $1/v$ with $v=15.0776$ and normalizing them.
\item We compute the confidence metric as the sum of the second and third largest weights before normalization, divided by two.
\end{enumerate}

Recall that the exponent value was chosen to maximise the weight difference between p-values $10^{-4}$ and $10^{-10}$ and that we experimentally observed that all choices of $v \in [10,30]$ give satisfactory results with low variance between them.

For each pair of sensors such that the corresponding weight is non-zero, we create a simulated TDOA measurement $\widehat{\Delta}_{ij}$ from the \emph{unknown} source to localize:
\begin{itemize}
\item we use the coordinates to compute the true time $\Delta_i$ taken by the signal to propagate from the unknown source to sensor $S_i$ for both sensors,
\item we compute the true TDOA $\Delta_{ij}=\Delta_i - \Delta_j$,
\item depending on the attack scenario, we add delays \\ $\Delta_{ij}'=\Delta_{ij} + a_i -a_j$,
\item we add Gaussian noise $\widehat{\Delta_{ij}}=\Delta_{ij}'+e_{ij}$ \\with $e_{ij} \in \mathcal{N}(0,2.192 e-9)$.
\end{itemize}
We implement the robust localization with the simulated measurements as in Algorithm~\ref{alg:localization}:
\begin{enumerate}
\item We compute a geometrical estimate of the source location $(x_g,y_g)$ as explained below.
\item We use the Matlab LM algorithm as a WLS estimator on all $\widehat{\Delta}_{ij}$ with weights $w_{ij}$ computed at step $(3)$. The initial step size is by default $0.01$ and the initial solution is $(x_g,y_g)$. We obtain the estimated robust source coordinates $(x,y)$.
\end{enumerate}
In all simulations, we analyze the confidence metric and the distance between our estimate and the true source coordinates. 
 
Recall that the WLS estimator we use is the LM algorithm. It is a gradient descent algorithm that requires an initial solution and step size that are updated iteratively. When the gradient is small, the step size is chosen small so that we can move gradually closer to the minima without missing it; in this case the algorithm is similar to the Gauss-Newton method. In contrast, when the gradient is large, the step size is chosen large and the algorithm behaves similarly to the steepest descent method. The initial solution we use is a geometrical estimate computed as the coordinate-wise weighted median of intersection points of all hyperbolae. The weight of both coordinates of an intersection point corresponds to the smallest weight among the weights of the corresponding TDOAs.

\subsection{Performance in Attack Scenarios}

In order to test the performance of our technique, we apply it in five different scenarios of attack with increasing attack delays. Each attack scenario corresponds to an attack location, specifically to a subset of the sensors. In all of the scenarios, we perform attacks with $25$ different delays ranging from $0$ to $50$ seconds. We simulated ten thousand times each attack scenario with each delay size. The $1'250'000$ results we obtained are presented in Figure~\ref{fig:precision}, where each color corresponds to a specific attack scenario. More specifically, Figure~\ref{fig:precision}a shows the confidence metric as a function of the distance between the true source position and our estimate in meters for each simulation. We refer to this distance as the estimate error. Figures~\ref{fig:precision}b and~\ref{fig:precision}c show the sample mean and confidence interval of the estimate error as a function of the delay size in seconds.  

The first scenario is a control scenario in which no attack takes place, it is presented in beige on the figures. We observe that the estimate error is on average below $0.5$m and that the confidence metric is always high, above $0.8$. 

Then, we consider a scenario where only sensor $S_1$ is under attack, it is presented in pastel green on the figures. We observe that whatever the delay injected, the estimate is always quite accurate with an error on average slightly above the one from the no-attack scenario. The estimate error is still in the proximity of $0.5$m and remains below $4$m. Figure~\ref{fig:precision}a shows that the confidence metric is also quite high as it remains above $0.7$. The slight reduction of estimate error comes from the fact that discarding signals from $S_1$ reduces the redundancy but still provides enough signals to locate accurately with one level of redundancy.

The third scenario, presented in purple, consists of attacking two sensors, $S_1$ and $S_2$, with the same delay. We observe that when the injected delay is below the standard deviation of the noise, the estimate error is as in the no-attack scenario. Then, as the delay increases, the distance between the estimate and the true source position also increases until it stabilises. This is explained by the fact that as the delay increases, the impacted TDOAs are less trusted but are still taken into account with a small weight, until they are completely discarded. At some point, only two TDOA values are trusted and used for localization. This is exactly enough, as the simulation is in 2D, but removes all redundancy. Therefore, the estimate grows less accurate. Nevertheless, the estimate error stays well below $4$m on average. This decrease in accuracy is accompanied by a decrease of the confidence metric value that is concentrated around $0.5$ and remains between $0.35$ and $0.7$. This illustrates that the source can be localized fairly correctly but with less accuracy as there is no redundancy.

The fourth scenario, presented in light blue, is performed by attacking $S_1$ and $S_2$ with the same delay of $500$s, such that only two TDOAs are used from the start of the simulation. Then, we increase slightly the delay difference between the sensors of synchronized pairs. Specifically, we attack in the following way:
\begin{itemize}
\item delay for $S_1$: $500$s,
\item delay for $S_2$: $(500+d)$s,
\item delay for $S_3$: $0$s,
\item delay for $S_4$: $(0+d)$s,
\end{itemize}
where $d$ is the delay difference that takes values from the same $25$ different delays considered above. We observe on Figure~\ref{fig:precision}c that the distance between the true source position and our estimate starts as in the previous scenario, which is as expected because only two TDOAs are trusted. Then, as the delay differences increase, the two TDOAs are increasingly affected, thus the localization relies solely on wrong TDOAs and the resulting estimate grows less accurate. We observe that when the delay differences are around $30n$s, the two TDOAs are declared as untrustworthy and the overall system as too corrupted to localize. In the worst case scenario, the estimate error is approximately $50$m but the corresponding confidence metric is around $2 \times 10^{-21}$, which is extremely low. For this simulation, Figure~\ref{fig:precision} shows that the estimates are less accurate and that the confidence metric is low, below $0.35$. Nevertheless, we can identify a grey zone, where the estimate error is below $4$m and the confidence metric is also low. In this case, a user of our solution might want to discard the estimate when in fact it is not very far from the true position of the source. Although this is unfortunate, it constitutes a false negative that is not as fatal as trusting a very inaccurate estimate.

We identified more cases of false negatives in the fifth scenario. In this setting, we start by attacking sensor $S_4$ with a delay of $500$s. In this way, the TDOAs with respect to $S_4$ are discarded from the start. Then, similarly to the fourth scenario, we increase slightly the delay differences between the three synchronized sensors:
\begin{itemize}
\item delay for $S_1$: $0$s,
\item delay for $S_2$: $(0+d)$s,
\item delay for $S_3$: $(0+2d)$s,
\item delay for $S_4$: $500$s.
\end{itemize}
The results of this scenario are depicted in coral red in Figure~\ref{fig:precision}. The accuracy of the estimate decreases in a fashion similar to the previous scenario but with lower error values. This is due to the fact that three TDOA values are used instead of two. Even though they are under attack, they include one level of redundancy. The largest distance between the source and the estimate remains below $10$m. After that, the system is declared as too corrupt. Although the accuracy is better in this scenario than in the previous one, the confidence metric ranges also between $0$ and $0.35$. Similarly, there are cases of accurate estimation but low confidence, which constitutes false negatives.

\begin{figure*}[t]

\includegraphics[width=\linewidth]{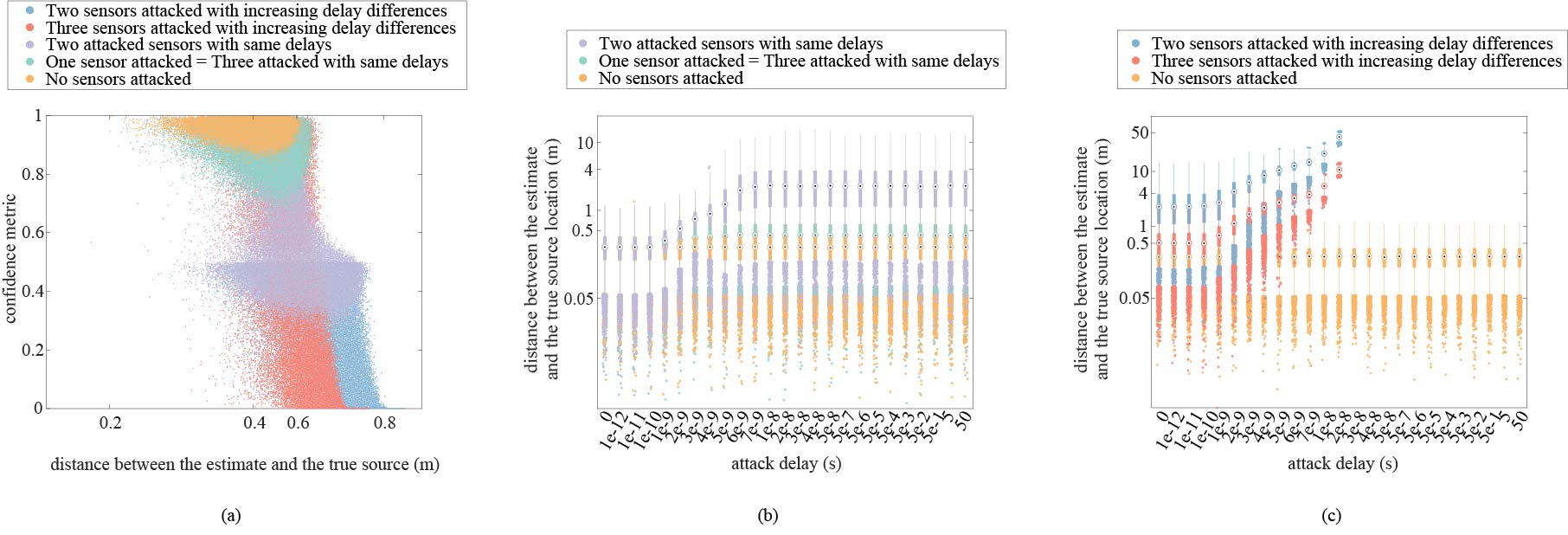}

\caption{Results from five different timing attack scenarios each with $25$ different delays, each simulated $10'000$ times: (a) confidence metric as a function of the distance between the true source position and the estimate provided by our robust solution, we observe that the metric is related to the accuracy and shows if there is redundancy in the measurements. There are some false-negative cases but no false positives. (b) the mean and confidence interval of the estimate error for each attack delay for three different scenarios. (c) the mean and confidence interval of the estimate error for each attack delay for two other scenarios: when the system is too corrupt, it stops localizing.}
\label{fig:precision}
\end{figure*}

In summary, our method gives an estimate which is always quite accurate, considering the fact that blindly trusting all TDOAs would lead to errors of many kilometers. When the system is too corrupt, it declares that the confidence level is too low to localize. We showed that the confidence metric gives useful insight on the accuracy of the estimate, although it can lead to some false-negative cases. We have not found any corner cases of false positives, in other words, our solution never trusts a highly inaccurate solution. If we were to recommend a course of action depending on the confidence metric, it would be the following:
\begin{itemize}
\item confidence metric $\in [0.75,1]$: trust the estimate to be as accurate as it can be because it includes at least one level of redundancy,
\item confidence metric $\in [0.3,0.75]$: probably computed with no redundancy, trust the estimate to be fairly correct but slightly less accurate,
\item confidence metric $\in [0,0.3]$: the true source is in a probable zone around the estimate, this result is not very accurate and trusting it depends on the application,
\item algorithm~\ref{alg:localization} output is \emph{"corrupt$\_$system"}: the attacks are too important to define even a probable zone of location. 
\end{itemize}

\subsection{Trajectory Simulations}

Next, we simulate the localization of an unknown source at nine different time instants and compare the true trajectory, the trajectory estimated with our robust solution, and the trajectory naively estimated by trusting all measurements. The naive estimates are found by the WLS estimator with all weights set to $1$. This section shows that our solution substantially improves the accuracy of the estimates when compared with the accuracy of the naive estimates obtained by ignoring the presence of timing attacks. We do so in three different attack scenarios, one for each confidence metric interval identified above.
 
In the first scenario we consider, only sensor $S_4$ is attacked with a fixed delay of $30 \mu $s. The results are given in Figure~\ref{fig:diff1}. Shown as out-of-bounds, the naive estimates are more than $100'000$ km away from the true source position at the first three time instants. Such obviously incorrect estimates would be flagged as bad data as it is not plausible to consider an estimate that far. Then, for the six remaining time instants, the naive estimates are off by more than two kilometers. In contrast, when we zoom closer to the source at the sixth time instant, we can observe that the distance between the true source and our robust estimate is always under $60$ cm. Note that the confidence metric is always above $0.86$ in this scenario.

\begin{figure}
	\centering
	\includegraphics[width=0.9\columnwidth]{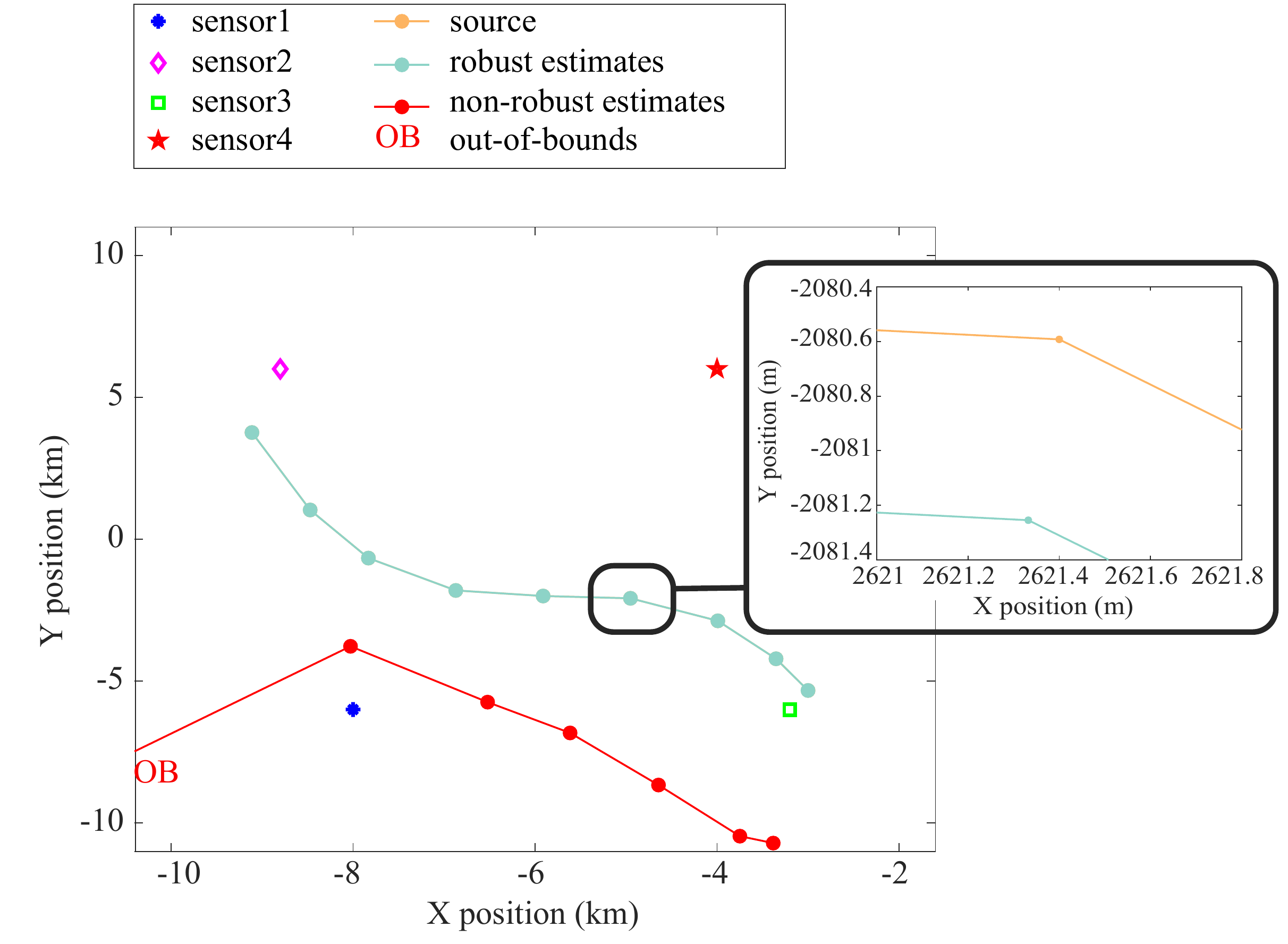}
	\caption{Estimation results at nine different time instants when $S_4$ is attacked by $30 \mu$s, the confidence metric is above $0.86$: the first three naive estimates are out-of-bounds and then wrong by $2$ km, whereas our solution provides estimates indistinguishable from the source; at the sixth time instant, our robust estimate is less than $60$cm away from the source.}
	\label{fig:diff1}
\end{figure}

In the second scenario, sensors $S_1$ and $S_3$ are both attacked with a fixed delay of $30 \mu $s. The results given in Figure~\ref{fig:diff2}, show that the estimate obtained naively is out-of-bounds at the first time instant. In fact, it is incorrect by more than $20$ km at first and by approximately $5$ km afterwards. In contrast, our solution provides estimates that are much more accurate, as they are all under $2$m away from the true position. In this scenario, the confidence metric is always above $0.38$.

\begin{figure}
	\centering
	\includegraphics[width=0.9\columnwidth]{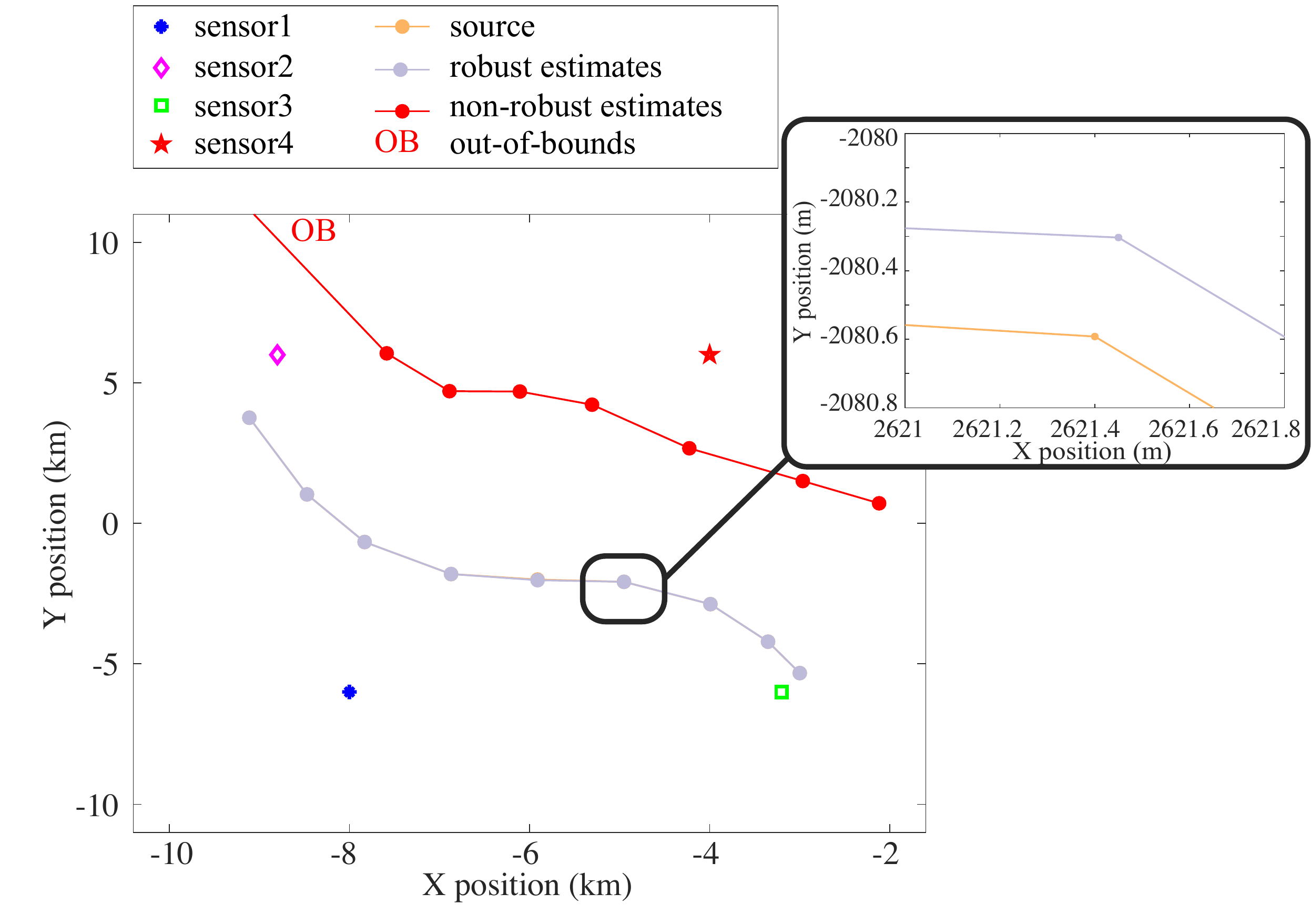}
	\caption{Estimation results at nine different time instants when $S_1$ and $S_3$ are attacked by $30 \mu$s, the confidence metric is above $0.38$: the naive estimate is out-of-bounds at first and then wrong by more than $5$ km, whereas our solution provides estimates indistinguishable from the source; at the sixth time instant, our robust estimate is less than $50$cm away from the source.}
	\label{fig:diff2}
\end{figure}

In the last trajectory scenario, we attack sensors $S_1$ and $S_2$ with a delay of $5$ seconds and we add a delay of $3n$s to $S_2$ and $S_4$. Namely, sensor pairs $(S_1,S_2)$ and $(S_3,S_4)$ are believed to be time synchronized when in fact, they have a delay difference slightly above the usual noise standard deviation $2.192n$s. In this scenario, the confidence metric is always well below $0.25$. Figure~\ref{fig:diff2d} shows that the naive estimates are always out-of-bounds as they are wrong by more than $1'000'000$ km. Figure~\ref{fig:diff2d} also shows that our estimate is always very close to the source. More precisely, our solution provides estimates that are always under $5$m away from the true source position, except for a corner case at the third time instant where our estimate is off by $15$m. 

\begin{figure}
	\centering
	\includegraphics[width=0.9\columnwidth]{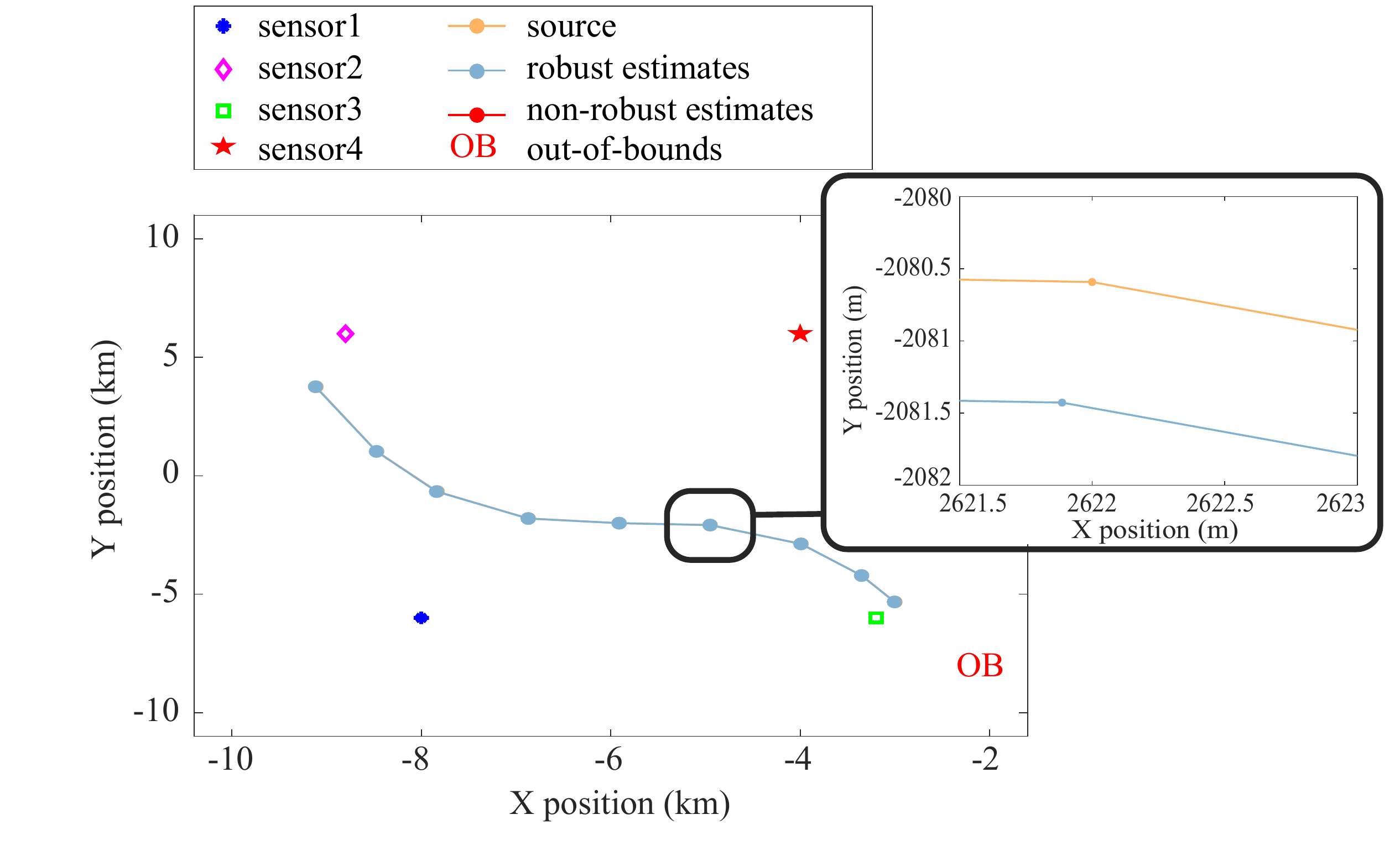}
	\caption{Estimation results at nine different time instants when $S_1$, $S_2$ and $S_4$ are respectively attacked with a delay of $5$s, $(5+3e-9)$s and $3e-9$s; the confidence metric is below $0.25$: the naive estimates are out-of-bounds at all time instants, whereas, our solution provides estimates indistinguishable from the source; at the sixth time instant, our robust estimate is less than $1$m away from the source.}
	\label{fig:diff2d}
\end{figure}

\subsection{Three-Dimensional Simulation}\label{sec:3D}

Lastly, we performed a three-dimensional trajectory simulation where sensor $S_1$ is attacked by $30 \mu$s as before. We added altitude coordinates $600$m, $1250$m, $900$m and $700$m to $S_1$, $S_2$, $S_3$ and $S_4$, respectively. In order to achieve the same level of redundancy as before, we placed a fifth sensor on the grid at an altitude of $400$ meters. From each TDOA, we computed the corresponding hyperboloids and used the coordinate-wise weighted median of intersection points as initial solution $(x_g, y_g,z_g)$ to the WLS estimator. The results are illustrated in Figure~\ref{fig:3D}. We observe in Figure~\ref{fig:3D} that the altitude of the naive estimates fluctuates far from the true altitude of the source. Whereas, our robust solution provides accurate estimates indistinguishable from the source. Similarly, Figure~\ref{fig:3D} shows that the 2D trajectory obtained with our solution on the $xy$-plane, matches with the trajectory of the source, whereas the naively estimated trajectory is always more than $2$km wrong. Finally, a close-up look in 2D at time instant number $4$ shows that our estimate is less than $2$m away from the source. However, not shown on the figures, at time instant number $4$, our estimate's altitude is off by approximately nine meters, whereas the naive estimate's altitude is off by more than a kilometer. Note that for this simulation, with weight function exponent $v=15.0776$, we obtain a high confidence metric equal to $0.94$.
In this simulation, the true source is at a constant altitude of $350$m, which is below the height of approximately the twenty highest towers on Earth. The minimal altitude found with the naive estimates is at $752$m, which is above all towers in the region of interest. A naive estimator would fail to detect a potential collision danger in this case. In contrast, the altitude of our robust estimates is always between $339$m and $377$m, which allows us to detect a dangerous flight behaviour.

\begin{figure}
\vspace{-0.5cm}
	\centering
	\includegraphics[width=0.8\columnwidth]{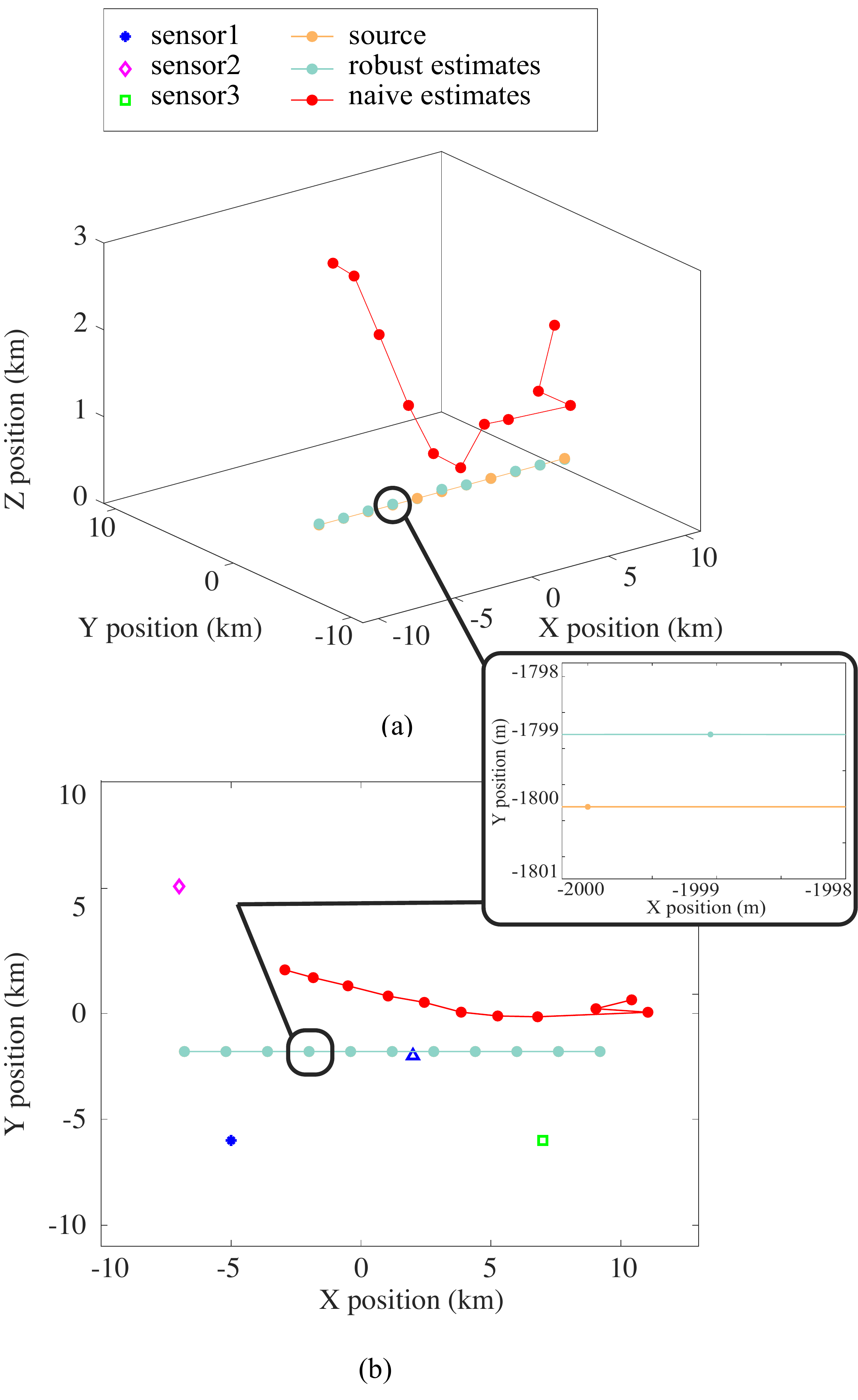}
	\caption{Estimation results in $3D$ at eleven different time instants when $S_1$ is attacked by $3e-5$s; the confidence metric is high at $0.94$: (a) The altitude of the naive estimates is very inaccurate at all time instants, our robust solution is indistinguishable from the source. (b) in the $xy$-plane, the naive estimates are more than $2$km wrong and our robust estimates are accurate. Zoom-in at time instant $4$: our robust estimate is less than $2$m away from the source on the $xy$-plane.}
	\label{fig:3D}
\end{figure}

\section{Conclusion}\label{sec:ccl}
To conclude, in this paper, we have shown that timing attacks on the time reference of the sensors of the network are a threat to TDOA localization. By injecting a few micro-seconds into the clock of a sensor, the network estimates the source to be located kilometers away from the true source position. We have also shown that a strong attacker with knowledge of the sensor coordinates and of the source coordinates, is able to choose the delays to inject such that the resulting misestimation results in a specific targeted location. To counter such timing attacks, we have proposed a robust technique that, to attribute weights to all sensor pairs, relies on signals from a known calibration source of a known position. These weights are computed to reflect the confidence we have in the time synchronization of the corresponding sensor pair. Subsequently, our localization technique uses these weights either to identify the network as too corrupt to localize, or to give an accurate estimate of the unknown source location. Our technique also provides a confidence metric that gives insight on the accuracy of the estimate. The calibration phase of our proposed solution is, however, vulnerable to replay attacks. In such attacks, the calibration signal is replayed at times and locations of the attacker's choice, thus possibly affecting the attributed weights. In order to counter these replay attacks, we have provided an encrypted authenticated challenge-response scheme that ensures that the measurements used for calibration are trustworthy.
Numerical evidence in 2D and 3D show that our technique is efficient and that the confidence metric is trustworthy although it might lead to false negatives.

\appendix
\section{Proofs}
{\small
According to Theorem~\ref{th:weak}, if the attacker is unable to jam signals, our encrypted authenticated challenge-response scheme is secure against a replay attack and the resulting $n$ measurements can then be given as input to Algorithm~\ref{alg:cal}. Now, in order to show that Claim~\ref{strong} is correct, we show that when an attacker is able to jam signals, successfully affecting $qm$ measurements, our strategy to select $n$ measurements is efficient, as long as $q \leq 0.45$. Specifically, we give numerical evidence that the weights that Algorithm~\ref{alg:cal} outputs from the $n$ selected measurements with and without a calibration attack are similar. We analyze results in two scenarios.

In the first, the measured TDOA value comes from a perfectly synchronized sensor pair. Hence, we require that Algorithm~\ref{alg:cal} on the $n$ selected measurements, outputs a weight close to $1$. In this case, the goal of the attacker is to jam signals in order to replay them with introduced delays such that Algorithm~\ref{alg:cal} outputs a weight close to $0$, thus making us discard trustworthy measurements. 
For every combination of calibration-attack size $a\in \{ 3,6,15000 \}$ and calibration-attack proportion $q \in \{ 0,0.1,0.2,0.3,0.4,0.45,0.5,0.6,0.7,0.8,0.9,1\}$, we performed the following procedure ten thousand times:
\begin{itemize}
\item we create a data set $DS_{160}$ of $m=160$ i.i.d samples where a proportion $q$ of the $m$ samples are drawn from $\mathcal{N}(\mu + a \sigma , \sigma)$ and the remaining are drawn from $\mathcal{N}(\mu, \sigma)$, with $\mu=7e-7$ and $\sigma=2.192e-9,$
\item we use our selection technique with $b=12$ bins, to extract the center of the highest Gaussian distribution and keep only the $n=30$ nearest samples, resulting in $DS_{30}$,
\item for both $DS_{160}$ and $DS_{30}$, we compute the p-value resulting from the z-test with $\mu=7e-7$ and $\sigma=2.192e-9$,
\item we apply the weight function of Algorithm~\ref{alg:cal} to the two p-values: we exponentiate them to the power $1/v$ with $v=15.0776$. 
\end{itemize}
At the end of the procedure, for every combination of $a$ and $q$, we obtain two data sets of ten thousand exponentiated p-values. Figure~\ref{fig:jammed} shows the sample mean of the resulting weights with a confidence interval for every $(a,q)$ pair, when using the entire $m=160$ measurements and when using the selected $n=30$ measurements. We observe that for a choice of $q \leq 0.45$, the weight obtained using the $30$ selected measurements is always high as in the no-attack scenario, in other words, as when $q=0$. Note that for $a=3$ and $q=0.45$, the weight is sometimes low although it is always large enough to ensure that the corresponding measurements are never incorrectly discarded. Such low weight values are due to the fact that the attack is small, comparable to large noise. Therefore the Gaussian distribution of the attacked measurements is merged with the Gaussian distribution of the unattacked measurements and the highest density peak is slightly shifted by the attack. When the attack is large enough for the intersection of the two Gaussian distributions to be empty, the attacked measurements are less likely to be mistaken for unattacked ones. Figure~\ref{fig:jammed} also shows that the weight obtained with all measurements without any particular strategy, decreases drastically as the size of the calibration attack increases. From this analysis, we observe that our strategy is very efficient in protecting our trust in synchronized sensor pairs.

\begin{figure}
	\centering
	\includegraphics[width=0.7\columnwidth]{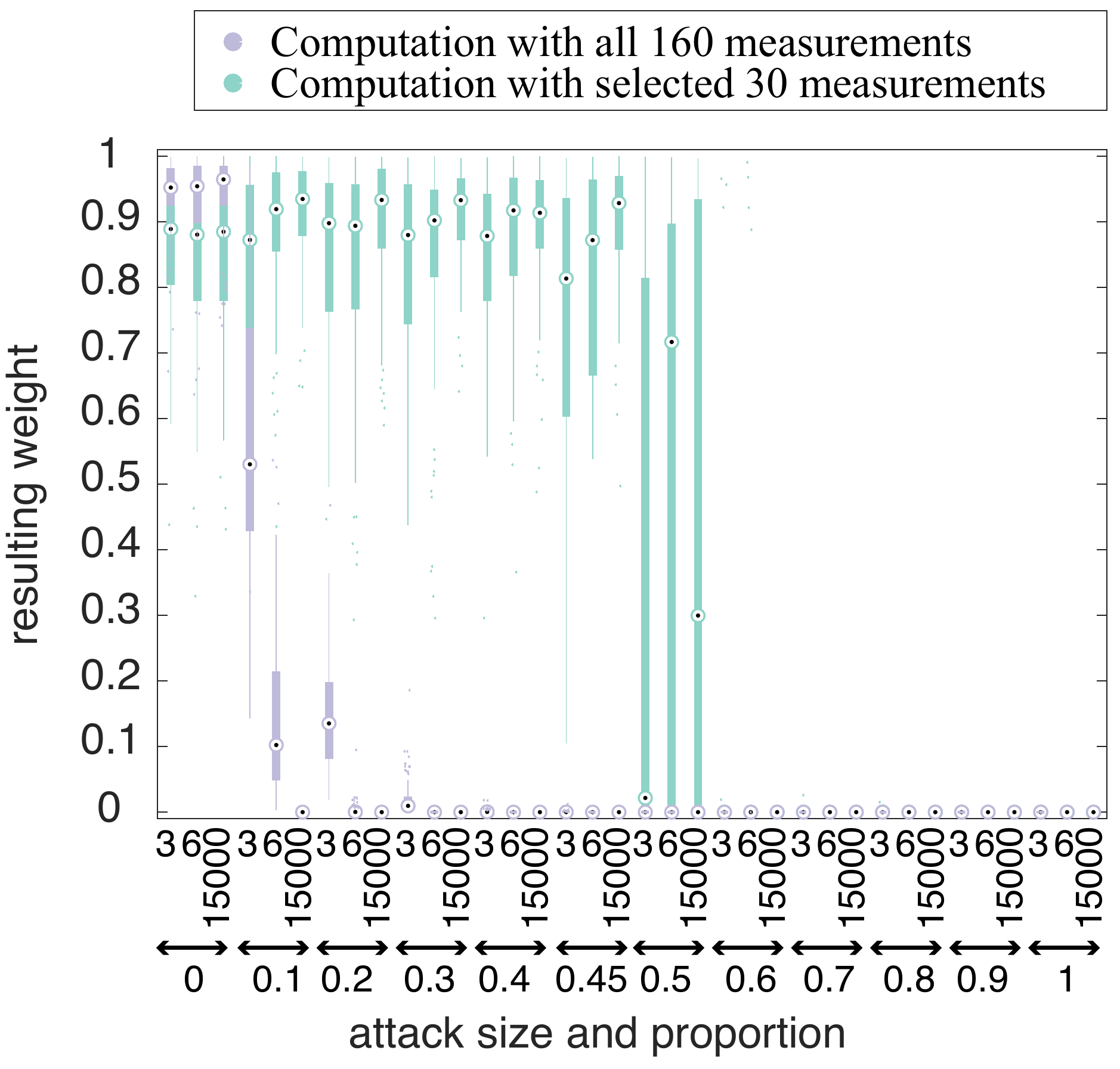}
	\caption{\small{Weight for an ideally synchronized sensor pair for different calibration-attack size and attack proportion from all $160$ and the selected $30$ measurements: when $q \leq 0.45$ the weight from the $30$ selected measurements is always high as in the no-attack scenario; the weight from all measurements decreases drastically as the size of the calibration attack increases.}}
	\label{fig:jammed}
\end{figure}

In the second scenario, the measured TDOA value comes from a sensor pair suffering from a timing attack. Hence, we require that Algorithm~\ref{alg:cal} outputs a weight close to $0$. In this case, the aim of the attacker is to attack calibration signals in order to compensate the synchronization delay between the attacked sensors. His goal is for Algorithm~\ref{alg:cal} to output a weight much higher than $0$, thus making us trust attacked measurements.
For every combination of timing-attack size $a\in \{ 3,6,15000 \}$ on the sensor pair, and calibration-attack proportion $q \in \{ 0,0.1,0.2,0.3,0.4,0.45,0.5,0.6,0.7,0.8,0.9,1\}$, we performed the following procedure ten thousand times:
\begin{itemize}
\item we create a data set $DS_{160}$ of $m=160$ i.i.d samples where a proportion $q$ of the $m$ samples are drawn from $\mathcal{N}(\mu  , \sigma)$ and the remaining are drawn from $\mathcal{N}(\mu+ a \sigma, \sigma)$, with $\mu=7e-7$ and $\sigma=2.192e-9,$
\item we use our selection technique with $b=12$ bins, to extract the center of the highest Gaussian distribution and keep only the $n=30$ nearest samples, resulting in $DS_{30}$,
\item for both $DS_{160}$ and $DS_{30}$, we compute the p-value resulting from the z-test with $\mu=7e-7$ and $\sigma=2.192e-9$,
\item we apply the weight function of Algorithm~\ref{alg:cal} to the two p-values: we exponentiate them to the power $1/v$ with $v=15.0776$. 
\end{itemize}
Figure~\ref{fig:jammed2} shows the sample mean of the resulting weights with a confidence interval for every $(a,q)$ pair, when using the entire $m=160$ measurements and when using the selected $n=30$ measurements. We observe that for a choice of $q \leq 0.45$, both the weights obtained using the $30$ and the $160$ measurements are low. From this analysis, we observe that when $q \leq 0.45$, our strategy is still very efficient in protecting our distrust in non-synchronized sensor pairs. 

\begin{figure}
	\centering
	\includegraphics[width=0.7\columnwidth]{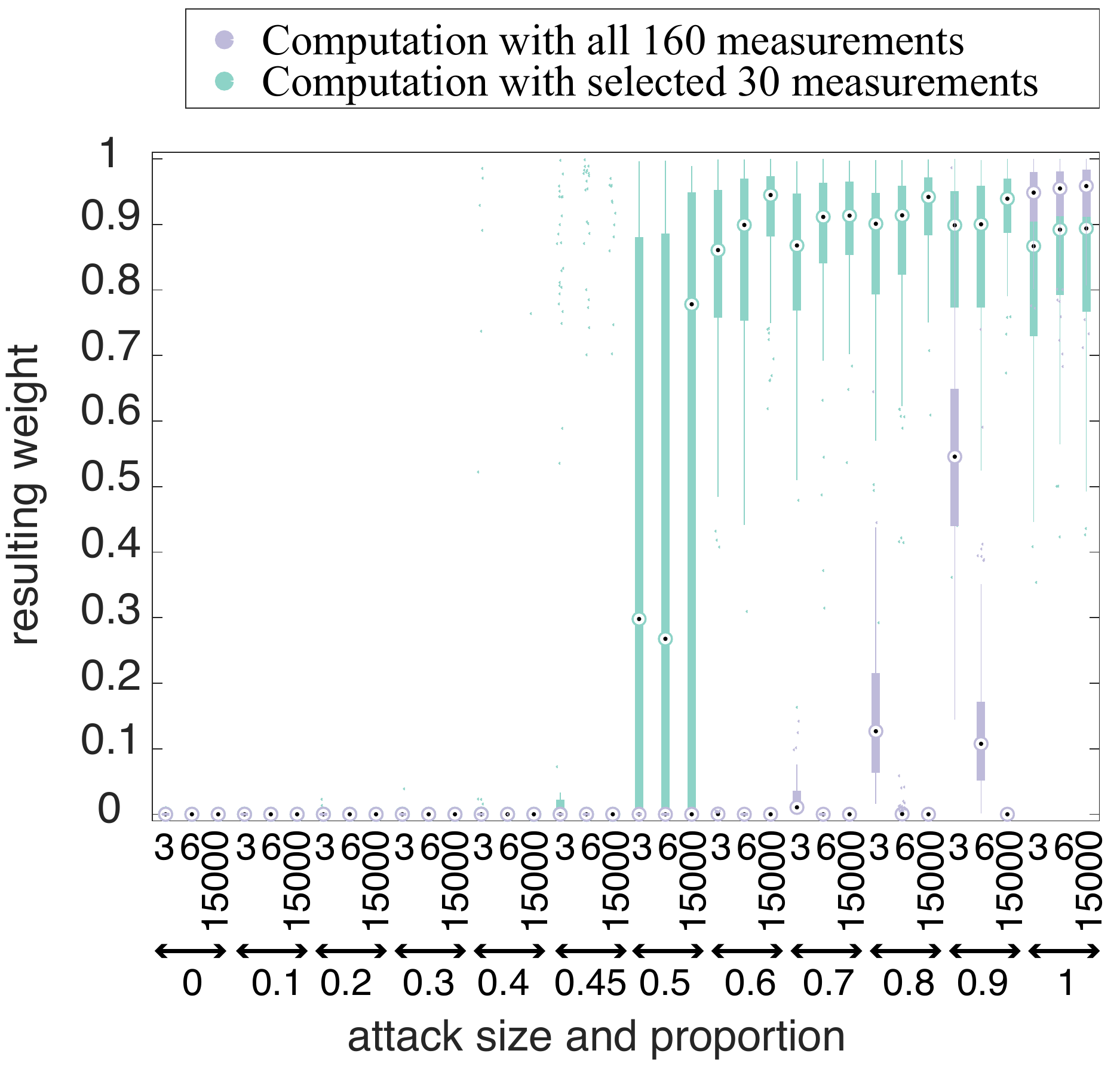}
	\caption{\small{Weight for a non-synchronized sensor pair for different timing-attack sizes and calibration-attack proportion from all $160$ and from the selected $30$ measurements: when $q \leq 0.45$, both the weights obtained using the $30$ and $160$ measurements are always low.}}
	\label{fig:jammed2}
\end{figure}

In summary, the numerical evidence shows that as long as $q \leq 0.45$, our technique to select $n$ out of the $m$ received measurements is efficient and enables Algorithm~\ref{alg:cal} to define proper weights. We further recall that if $q$ is selected too large, the attacker would be detected by other techniques such as SNR analysis.
}

\bibliographystyle{IEEEtran}
\bibliography{references-11}

\end{document}